\documentclass[conference]{IEEEtran}
\usepackage{times}

\usepackage{amsmath}
\usepackage{amsthm}
\usepackage{amssymb}
\usepackage[utf8]{inputenc}
\usepackage[T1]{fontenc}
\usepackage{graphicx}
\usepackage{enumerate}
\usepackage{xcolor}
\usepackage{comment}
\usepackage[numbers,sort&compress]{natbib}
\usepackage{multicol}
\usepackage{hyperref}
\usepackage{tabularx}
\usepackage{makecell}
\usepackage{colortbl}
\usepackage{tablefootnote}
\usepackage{multirow}
\usepackage{paralist}
\usepackage{placeins}
\usepackage[ruled,vlined,linesnumbered]{algorithm2e}

\usepackage[caption=false]{subfig}

\hypersetup{
    colorlinks=true,
    linkcolor=blue,
    filecolor=magenta,      
    urlcolor=cyan,
}

\newtheorem{theorem}{Theorem}
\newtheorem{proposition}{Proposition}
\newtheorem{lemma}{Lemma}

\theoremstyle{definition}

\IEEEoverridecommandlockouts

\begin{document}


\title{Gaussian Process Barrier States for  Safe Trajectory Optimization and Control}







\author{Hassan Almubarak, Manan Gandhi, Yuichiro Aoyama, Nader Sadegh, and Evangelos A. Theodorou$^{*}$
\thanks{$^{*}$Georgia Institute of Technology, Atlanta, GA, USA}}

\maketitle

\begin{abstract}
This paper proposes embedded Gaussian Process Barrier States (GP-BaS), a methodology to safely control unmodeled dynamics of nonlinear system using Bayesian learning. Gaussian Processes (GPs) are used to model the dynamics of the safety-critical system, which is subsequently used in the GP-BaS model. We derive the barrier state dynamics utilizing the GP posterior, which is used to construct a safety embedded Gaussian process dynamical model (GPDM). We show that the safety-critical system can be controlled to remain inside the safe region as long as we can design a controller that renders the BaS-GPDM's trajectories bounded (or asymptotically stable). The proposed approach overcomes various limitations in early attempts at combining GPs with barrier functions due to the abstention of restrictive assumptions such as linearity of the system with respect to control, relative degree of the constraints and number or nature of constraints. This work is implemented on various examples for trajectory optimization and control including optimal stabilization of unstable linear system and safe trajectory optimization of a Dubins vehicle navigating through an obstacle course and on a quadrotor in an obstacle avoidance task using GP differentiable dynamic programming (GP-DDP). The proposed framework is capable of maintaining safe optimization and control of unmodeled dynamics and is purely data driven.
\end{abstract}

\IEEEpeerreviewmaketitle

\section{Introduction}
Autonomy has increasingly been in demand with recent advances in robotic technologies and applications in automobile, healthcare, agriculture, and other fields to increase efficiency, safety, consistency and to lower costs. Undoubtedly, safety in modern robotic systems is a critical task in designing successful autonomous systems. Nonetheless, with the increasing demand for autonomy, designing safe and effective algorithms is an increasingly challenging task, especially in unknown or dynamic environments. Classical methods to safe autonomy depend on model based methods which require high knowledge of the physics of the overall system, e.g. the autonomous system and its environment. For many autonomous systems, however, developing such models is not attainable. Data-driven representations of dynamic systems have been shown to be capable of providing highly accurate dynamic models. Combining such methods with model based safety certificates and planning algorithms can provide great practical solutions. Gaussian processes (GPs) regression has emerged as an efficient Bayesian supervised learning tool \cite{rasmussen2003gaussian} and has been considered extensively in forecasting \cite{girard2002gaussian} and in the robotics and control literature for efficient learning \cite{wang2005gaussian, deisenroth2013gaussian, marco2016automatic, umlauft2017feedback, beckers2016stability, castaneda2021gaussian, Pan2018}. In this work, we develop a novel approach of incorporating a model-based barrier method, namely barrier states, with Gaussian processes to design a data driven multi-objective optimization and control of unknown dynamics.  

\begin{figure} [t]
        \centering
        \vspace{-18mm}
    \subfloat{\includegraphics[trim=50 70 50 80, clip, width=0.85\linewidth]{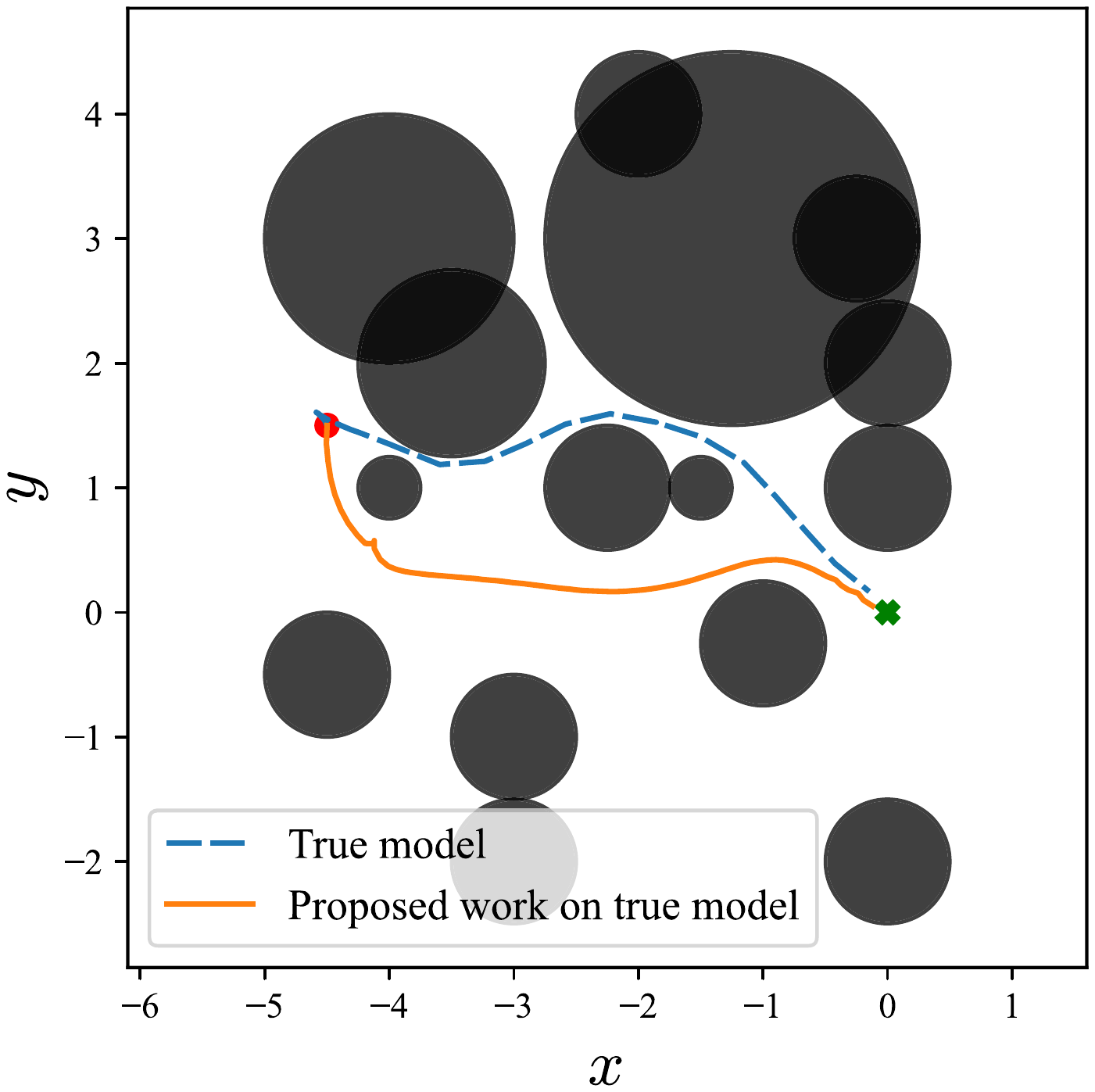}}
        \\ \vspace{-18mm} \hspace{-3mm}
    \subfloat{\includegraphics[trim=30 50 50 200, clip, width=0.65\linewidth]{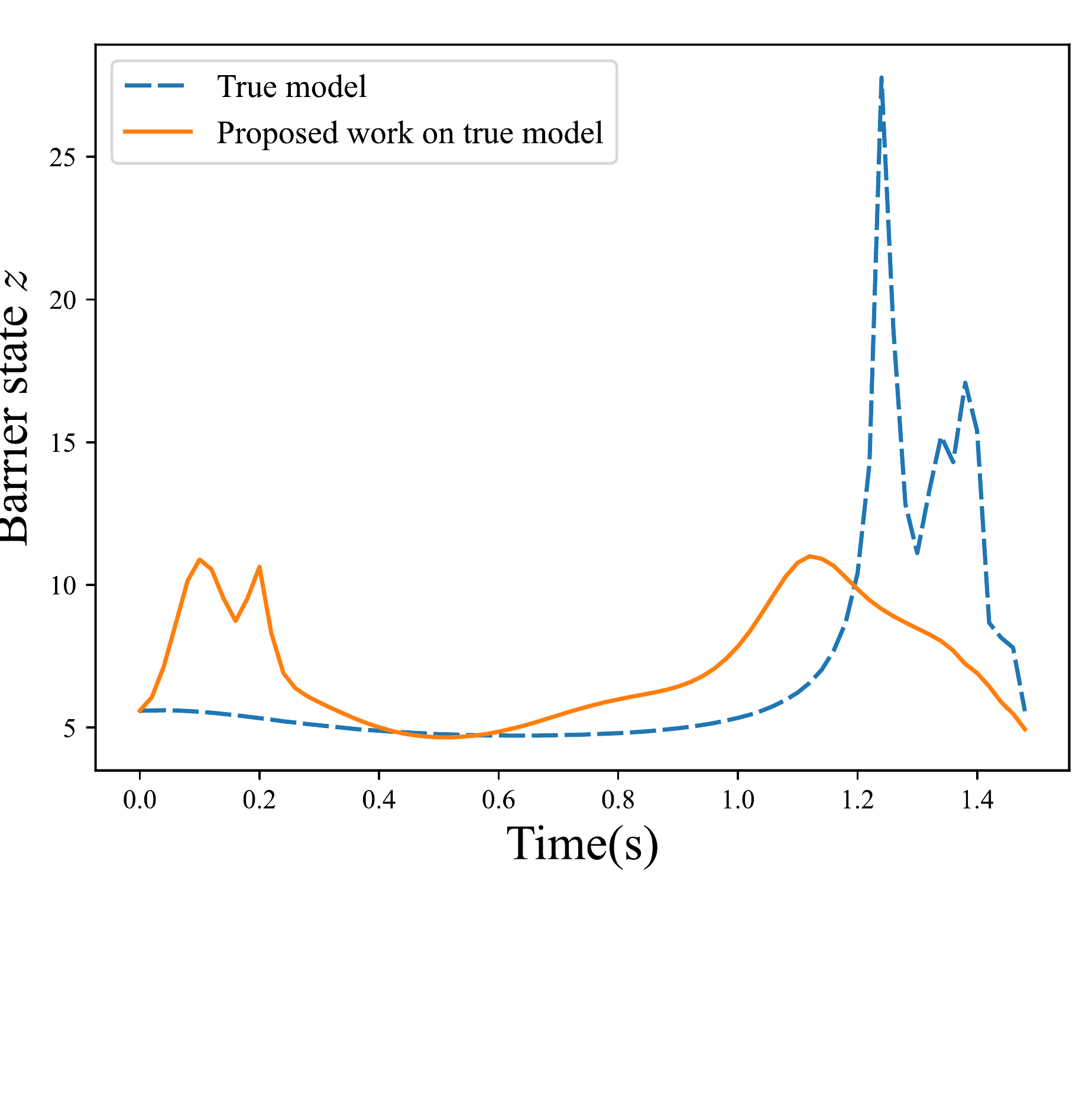}}
    \vspace{-15mm}
   \caption{Closed-loop trajectories of a Dubins robot navigating an obstacle course using true model (blue) and using the proposed framework (orange). The learned GP dynamics with the constructed GP-BaS used with GP-DDP \cite{pan2015data} to solve the problem in a purely data driven fashion. As shown, using the true system with BaS-DDP \cite{almubarak2021safeddp}, the robot is able to navigate in narrow areas while the proposed data-driven takes a more conservative path to stay away from the obstacles given uncertain dynamics.}
      \label{fig:dubins- many obst}
      \vspace{-6mm}
\end{figure} 
    
\subsection{Related Work}
To account for uncertain or unknown systems' models within barriers based safety-critical control, GPs have been used with barrier certificates. \citet{Wang2018} used barrier certificates to keep a quadrotor in a partially unknown environment within known safe regions and expand the safe regions as more data are collected and used to update a GP model of the system. In \cite{Jagtap2020}, GPs were used to learn the drift term of the model offline and provide Control Barrier Functions (CBFs) synthesis with estimated confidence bounds that depend on the validity of the learned model using Monte-Carlo simulations. 
This work considers a relative degree one CBF synthesis for which the input gain matrix is known. \citet{dhiman2021control} proposed a more general framework using matrix variate GPs models to learn the unknown system's dynamics, drift term and input gain, online for systems linear in control to enforce stability and safety constraints. Nonetheless, although multiple methods including covariance vectorization are used, the main pitfall of the approach comes from the fact that to learn the full dynamics $n(1+m)$ independent GP models need to be learned which inflates the problem's computational complexity. The relative degree of the constraint is also assumed to be known. \citet{castaneda2022probabilistic} considers a semi-model based approach in which a nominal model is exploited to be used in CBF while collecting data online in an event-triggered fashion to guarantee recursive feasibility of safe controls of the unknown system given a relative degree one CBF with respect to the system. 

Barrier states (BaS) embedding methodology has emerged recently as a viable method for enforcing safety in multi-objective control and optimization \cite{Almubarak2021SafetyEC,almubarak2021safeddp}. This method enforces safety through embedding the dynamics of the barrier function into the model of the safety-critical system effectively transforming the safety constraints into performance objectives of a new fictitious system. As a barrier based method, barrier states need the model of the system to guarantee safety. When an uncertain model is used, safety guarantees are no longer applicable. In \cite{almubarak2021safeminmax}, a min-max optimization approach was used to provide robustness against model uncertainty. Nonetheless, safety guarantees depend on the feasibility of the dynamic game's solution which may not always exist. In this work, we utilize Gaussian processes to obtain a posterior distribution of the unknown dynamics to construct Gaussian process barrier states (GP-BaS) for purely data-driven multi-objective safe optimizations and controls. 

\subsection{Contributions and Organization}
The main contribution of the proposed work lies in using GPs as a Bayesian learning tool to construct data-driven probabilistic safety guarantees through barrier states and design safe optimal control and planning. As a result of using BaS with GPs, various advantages over related work are gained. This includes, avoiding the assumptions on the linearity of the system with respect to the control which facilitates using developed efficient learning schemes to learn the fully unknown dynamics. Additionally, further assumptions on the relative degree of the barrier function with respect to the system and designing a barrier function for each constraint, relaxation between safety and performance objectives, are avoided as well. Moreover, we consider the problem of trajectory optimization of the unknown system. To achieve our objectives, we derive Gaussian Barrier states (GP-BaS) dynamics that are used to derive probabilistic safety guarantees. The GP-BaS is used to construct a safety embedded Gaussian process dynamical model (GPDM) that is used for control design. Consequently, if this model's states trajectories are rendered bounded, then safety of the safety-critical system is guaranteed in probability. Furthermore, we use uncertainty propagation for safe trajectory optimization within a differentiable dynamic programming (DDP) framework. 

The paper is organized as follows. In  \autoref{Sec:Problem Statement}, we start with the problem statement and briefly review barrier states for known dynamics. \autoref{Section: Bayesian Barrier States} contains the main contribution of the work in which GP-BaS dynamics and the safety embedded Gaussian system are derived. In this section, probabilistic safety guarantees based on the derivations of GP-BaS are provided along with using uncertainty propagation for trajectory optimization. In \autoref{Section: Applications and Examples}, we implement the proposed paradigm on various systems including optimal safe control of an unstable linear system as a proof of concept, and trajectory optimization of Dubins vehicle in which we show that the proposed approach can produce a more conservative solution as in \autoref{fig:dubins- many obst} and a quadrotor navigating an obstacles course using GP-DDP \cite{pan2015data}. Finally, concluding remarks and future directions are provided in \autoref{Section: Conclusion}.

\section{Problem Statement} \label{Sec:Problem Statement}
Consider the general \textbf{unknown} nonlinear control system
\begin{equation} \label{eq:dynamics}
    \dot{x}= f(x,u)
\end{equation}
where $t \in \mathbb{R}^+$, $x \in \mathcal{X} \subset \mathbb{R}^n$, $u \in \mathcal{U} \subset \mathbb{R}^m$ and $f:\mathcal{X} \times \mathcal{U} \rightarrow \mathcal{X}$ is continuously differentiable.
We wish to design a \textit{safe feedback} controller $u(x)$ that achieves performance objectives and observe the safety condition $h(x)>0$ where $h(x)$ is a continuously differentiable function that describes the safe set $\mathcal{S}=\{x|h(x)>0\}$. The performance objectives we consider in this paper include enforcing stability, in which we assume $f(0,0) = 0$ without loss of generality, as well as trajectory planning. As the dynamics of the system is unknown, we use Gaussian processes regression to learn a GP model for $f$ given $x$, $u$ and $\dot{x}$. 

\subsection{Embedded Barrier States for Multi-objective Control of Known Dynamics}
Embedded barrier states emerged as a model-based methodology that allows us to construct multi-objective control for safety-critical system without relaxation of one of the objectives. This technique is flexible enough to be integrated with existing control techniques for general class of nonlinear systems and for general nonlinear constraints \cite{Almubarak2021SafetyEC,almubarak2021safeddp}. In principle, those barrier states (BaS) create barriers in the state space of a new system forcing the search of a multi-objective feedback control law to the set of safe controls. Specifically, the barrier states are appended to the model of the safety-critical dynamical system to generate a system that is safe if and only if its trajectories are rendered bounded. In the case of stabilization, for example, safety is guaranteed if the augmented system's equilibrium point of interest is asymptotically stabilized. Therefore, since safety and stability have been unified, designing a stabilizing controller for the new model means a stabilizing safe control for the original dynamical system \cite{Almubarak2021SafetyEC}. In other words, the safety property is embedded in the stability of the system as the feedback controller is a function of the system's states and the barrier states, hence the name \textit{safety embedded control}. The proposed method is general enough to be used with any valid barrier function and for general class of nonlinear systems and constraints. Unlike control barrier functions (CBFs), no explicit knowledge of the relative degree of the system with respect to the output describing the safe region is required and therefore the constraints can be directly used to construct the barriers.

For the nonlinear system \eqref{eq:dynamics} and the safe set $\mathcal{S}$ defined by the safety function $h(x)$, define a barrier function $\beta(x)=B(h(x))$. The idea is to augment the open loop system with a new state variable $z:= \beta(x) - \beta(0)$ to ensure that $z=0$ is an equilibrium state. Then, the barrier state equation can be given as 
\begin{equation*}
 f^z := \dot{z} = B'\big(B^{-1}(z+\beta_0) \big) L_{f(x,u)}h(x) - \gamma \big(z+\beta_0 - B(h(x)) \big)
\end{equation*}
for $\gamma \in \mathbb{R}^+_0$. Note that shifting the equilibrium point of the barrier state is not necessary but ensures that the origin is an equilibrium state of the augmented system as well. For notational convenience, we define $\mathcal{B}(x,z):=B'\big(B^{-1}(z+\beta_0) \big) h_x(x)$ and thus the BaS dynamics can be written as
\begin{equation} \label{eq:zdot-barrier state}
 \dot{z} = f^z(x,z,u) = \mathcal{B}(x,z) f(x,u) - \gamma \big(z+\beta_0 - B(h(x)) \big)
\end{equation}


It is worth noting that multiple constraints can be combined to construct a single BaS \cite{Almubarak2021SafetyEC,almubarak2021safeddp}. This is of a great advantage instead of stacking multiple constraints especially for applications that involve a large number of constraints, e.g. in collision avoidance in cluttered environments. Nonetheless, in some applications, one may want to construct multiple barrier states to decouple the treatment of the different constraints. Therefore, for the BaS vector $z=[z_1, \dots, z_q]^{\mathsf{T}} \in \mathcal{Z} \subseteq \mathbb{R}^q$ where $q$ is the number of barrier states, the augmented system is given by 
\begin{equation} \begin{split} \label{eq:new system, safety augmented}
    & \dot{\bar{x}}= \bar{f}(\bar{x}, u) \\ 
\end{split} \end{equation}
where $\bar{x}=\begin{bmatrix} x \\ z \end{bmatrix}, \bar{f}=\begin{bmatrix} f \\ f ^z\end{bmatrix}$ with $\bar{f}(0)=0$ preserves the continuous differentiability and stabilizability of the original control system \eqref{eq:dynamics}. Therefore, the safety constraint is \textit{embedded} in the closed-loop system's dynamics and stabilizing the safety embedded system \eqref{eq:new system, safety augmented} implies enforcing safety for the safety-critical system \eqref{eq:dynamics}, i.e. forward invariance of the safe set $\mathcal{S}$ with respect to \eqref{eq:dynamics}.
\begin{proposition}[\cite{Almubarak2021SafetyEC}] \label{prop: safety if bas is bounded}
A continuous feedback controller $u=K(x)$ is said to be safe, i.e. the safe set $\mathcal{S}$ is controlled invariant with respect to the closed-loop system $\dot{x}=f\big(x,K(x)\big)$, if and only if $z(0)<\infty \Rightarrow z(t)<\infty \ \forall t$.
\end{proposition} 
This implies that boundedness of $z$ guarantees satisfaction of the safety constraint. As a consequence, the forward-invariance of $\mathcal{S}$, i.e. safety of the safety-critical system, can be tied to the performance objectives of the safety embedded system \eqref{eq:new system, safety augmented}. In other words, boundedness of the BaS implies the generation of safe trajectories \cite{almubarak2021safeddp}. For stabilization objectives specifically, rendering the augmented systems asymptotically stable implies enforcing asymptotic stability as well as safety for the safety-critical system as shown by the following theorem. 
\begin{theorem}[\cite{Almubarak2021SafetyEC}] \label{theorem: original safe if stable theorem}
Assume there exists a continuous feedback controller $u=K(\bar{x})$ such that the origin of the safety embedded closed-loop system, $\dot{\bar{x}}=\bar{f}\big(\bar{x}, K(\bar{x})\big)$, is asymptotically stable. As a result, there exists an open neighborhood $\mathcal{D}$ of the origin such that $u=K(\bar{x})$ is safe, i.e. $x(t) \in \mathcal{S} \ \forall t$. 
\end{theorem}

\section{Bayesian Barrier States} \label{Section: Bayesian Barrier States}
As the embedded barrier states method is a model-based approach, the unknown system's dynamics $f(x,u)$ is needed in the BaS equation for control design and in the propagation of the BaS to be used in the feedback controller for safety guarantees.

Viewing state transitions as an inference problem, we use GPs to learn a dynamical model of the system. In essence, we assume that system's transition for a specific instant in time can be approximated with a Gaussian random variable, effectively approximating the dynamics with independent Gaussian processes. A huge advantage of GPs is that a GP is completely defined by its mean and variance. We use GPs predictions to provide mean and variance for $f$, i.e. we have probabilistic estimations. Consequently, the BaS is now a random variable and safety is probabilistic in general. Given a prior distribution of $f(x,u)$, mean and variance, we wish to compute its posterior distribution using state and control observations. Note that in this step, we assume availability of $\dot{x}$, which can be approximated using historical data. Our objective is to satisfy the safety and performance objectives while mitigating the impact of estimation error.

Let $\hat{x}:=(x,u)$ be a pair of state and control and let us define the series of $N$ data points $\hat{X}:=\{ \hat{x}_0, \dots, \hat{x}_N\}$ and the corresponding transition of the dynamics
$F:=\{f_0, \dots, f_N\}$. 
Given independent identically distributed (iid) outputs, with test data $\hat{x}^*:=(x^*,u^*)$ and assuming zero mean process, without loss of generality, with variance $\sigma^2_n$, the GP regression technique uses marginalization to predict the joint distribution of the output of the observed data and the output of the test data. More specifically, 
\begin{equation}
    p \begin{pmatrix} F \\ f^* \end{pmatrix} \sim \mathcal{N} \Big( 0, \begin{bmatrix} {\rm{K}}(\hat{X},\hat{X}) + \sigma^2_n I & {\rm{K}}(\hat{X},\hat{x}^*) \\
    {\rm{K}}(\hat{x}^*, \hat{X}) & {\rm{K}}(\hat{x}^*, \hat{x}^*) \end{bmatrix} \Big)
\end{equation}
where the covariance is defined by the kernal matrix ${\rm{K}}$, which can be viewed as a similarity measure. This is used to get the predictive mean and the predictive variance of the GP model of $f$, also known as the posterior distribution, as
\begin{align} \label{eq:system's Gaussian model}
 {\mathbb{E}}[f] & = \mu(\hat{x}) = {\rm{K}}(\hat{x}^*,\hat{X}) \mathbf{K}^{-1} F \\
 \mathbb{V}[f] & = \Sigma(\hat{x}) = {\rm{K}}(\hat{x}^*,\hat{x}^*) - {\rm{K}}(\hat{x}^*, \hat{X}) \mathbf{K}^{-1} {\rm{K}}(\hat{X},\hat{x}^*)
\end{align}
where $\mathbf{K}={\rm{K}}(\hat{X},\hat{X}) + \sigma^2_n I$ and $\mu$ and $\Sigma$ represent the Gaussian distribution of $f \sim \mathcal{N} (\mu, \Sigma)$.

\subsection{Gaussian Process Barrier States (GP-BaS)}
Using the definition of the GP model of the dynamical system and the definition of barrier states in \autoref{Sec:Problem Statement}, we can develop the following results.
\begin{lemma} \label{lemma: gaussian bas lema}
For the unknown dynamical system \eqref{eq:dynamics} with the Gaussian model \eqref{eq:system's Gaussian model}, the barrier states dynamics, $f^{z}$, are Gaussian processes and thus can be completely defined by:
\begin{align} \label{eq:E[fz]}
     & {\mathbb{E}}[f^{z}|\hat{x}]  =  \mu^z = \mathcal{B}(x,z) {\mathbb{E}}[f(x,u)|\hat{x}]  - \gamma \big(z+\beta_0 - B(h(x)) \big)
     \\ \label{eq:Var[fz]}
     & \mathbb{V}[f^{z}|\hat{x}]  =  \Sigma^z = \mathcal{B}(x,z)  \mathbb{V}[f(x,u)|\hat{x}] \mathcal{B}(x,z)^{\mathsf{T}}
\end{align}
Those equations represent the \textbf{Gaussian Process Barrier States (GP-BaS)}, $f^z \sim \mathcal{N} (\mu^z, \Sigma^z)$. 
\end{lemma}
\begin{proof}
For the barrier states dynamics in \eqref{eq:zdot-barrier state} and the GP model in \eqref{eq:system's Gaussian model}, we have its mean to be
\begin{equation*}
    \begin{split}
        \mathbb{E}[f^{z}|\hat{x}] = & \mathbb{E} \big[B'\big(B^{-1}(z+\beta_0) \big) h_x f(x,u) \\ 
      &- \gamma \big(z+\beta_0 - B(h(x)) \big) \big|\hat{x} \big] \\
      = & \mathcal{B}(x,z) {\mathbb{E}}[f(x,u)|\hat{x}]  - \gamma \big(z+\beta_0 - B(h(x)) \big)
    \end{split}
\end{equation*}
For notational convenience, the arguments are dropped unless necessary. Then, the variance is given by
\begin{equation*}
\resizebox{1.03\hsize}{!}{$
    \begin{split}
        \mathbb{V}[f^{z}] = & \mathbb{E} \Big[ \big(f^z - \mathbb{E}[f^z ] \big) \big(f^z - \mathbb{E}[f^z ] \big)^{\mathsf{T}}  \Big] \\ 
        = & \mathbb{E} \Big[ \Big( \mathcal{B} f - \gamma \big(z+\beta_0 - B(h) \big) -  \mathcal{B} {\mathbb{E}}[f]  + \gamma \big(z+\beta_0 - B(h) \Big) \\
        & \Big( \mathcal{B} f - \gamma \big(z+\beta_0 - B(h) \big) -  \mathcal{B} {\mathbb{E}}[f]  + \gamma \big(z+\beta_0 - B(h) \Big)^{\mathsf{T}}  \Big] \\
        = & \mathbb{E} \Big[  \mathcal{B} \big( f - \mathbb{E}[f ] \big)  \big( f - \mathbb{E}[f ] \big)^{\mathsf{T}} \mathcal{B}^{\mathsf{T}}  \Big]   \\
        = & \mathcal{B} \ \mathbb{E} \Big[\big( f - \mathbb{E}[f ] \big)  \big( f - \mathbb{E}[f] \big)^{\mathsf{T}} \Big]  \mathcal{B}^{\mathsf{T}}  = \mathcal{B}  \mathbb{V}[f] \mathcal{B}^{\mathsf{T}}
    \end{split}$}
\end{equation*}
which completes the proof.
\end{proof}

Notice that the GPs posterior so far is a one-step prediction model and therefore the proposed GP modeled BaS here does not account for propagation of uncertainty and assumes access to $\hat{x}$, i.e. states measurement $x$ and control signal $u$. For trajectory optimization, it must be noted that the GPs need to be propagated forward in time and hence multiple-step ahead prediction is needed.

\subsection{Gaussian Process Barrier States (GP-BaS) In Trjaectory Optimization}
In trajectory optimization, long-term predictions are needed in which uncertain inputs are used. This means that the uncertainty of a prediction is compounded when input back into the GP dynamics which also means that a probability distribution is propagated through the nonlinear kernel function. As a result, the posterior's distribution is now non-Gaussian with intractable integral. This calls for approximations of the GPs predictions. One could use a Monte-Carlo approximation, moment matching or Taylor series expansion of the posterior GP mean function. We resort to the latter solution, which is effectively an approximation of a Gaussian model, due to its computational advantage \cite{girard2002gaussian,deisenroth2013gaussian}.

Following the derivations in \cite{girard2002gaussian}, we consider the first order Taylor expansion of 
of the posterior around the GP mean $\mu_{\hat{x}}$, where the subscript indicates the evaluation under $\hat{x}$:
\begin{equation*}
    \mu^z(\mu_{\hat{x}}, \Sigma_{\hat{x}}) \approx \mu^z(\mu_{\hat{x}}) + \frac{\partial \mu^{z}}{\partial \hat{x}} \Big|_{\hat{x}=\mu_{\hat{x}}}   (x-\mu_{\hat{x}}) 
\end{equation*}
and therefore
\begin{equation}
    \mathbb{E}[f^{z} | \mu_{\hat{x}}, \Sigma_{\hat{x}}] =  \mathbb{E} [\mu^z(\mu_{\hat{x}}, \Sigma_{\hat{x}})] \approx \mu^z(\mu_{\hat{x}})
\end{equation}
Similarly for the variance,
\begin{equation*}
    \Sigma^z(\mu_{\hat{x}}, \Sigma_{\hat{x}}) \approx \Sigma^z(\mu_{\hat{x}}) + \frac{\partial \Sigma^{z}}{\partial \hat{x}} \Big|_{\hat{x}=\mu_{\hat{x}}} (x-\mu_{\hat{x}}) 
\end{equation*}
leading to
\begin{equation}
    \mathbb{V}[f^{z} | \mu_{\hat{x}}, \Sigma_{\hat{x}}] \approx  \Sigma^z(\mu_{\hat{x}}) + \frac{\partial \mu^{z}}{\partial \hat{x}} \Big|_{\hat{x}=\mu_{\hat{x}}}^{\mathsf{T}} \Sigma^z_{\hat{x}} \frac{\partial \mu^{z}}{\partial \hat{x}} \Big|_{\hat{x}=\mu_{\hat{x}}}
\end{equation}


\subsection{Probabilistic Safety Guarantees Using GP-BaS}
For the learned dynamics, we can have a high probability confidence interval of the unmodeled dynamics $f(x,u)$:
\begin{equation}  \label{eq:confidence interval set for f}
     {\mathbb{E}}[f] - \phi_\rho \sigma[f] \leq f \leq {\mathbb{E}}[f] + \phi_\rho \sigma[f]
\end{equation}
where $\sigma[\cdot]$ is the standard deviation and $\phi_\rho$ is a user selected variable that determines the confidence level. Similarly, for the barrier states, given a desired safety probability $\rho \in (0, 1)$, we can define the quantile function
\begin{equation}
    \Phi^{-1}(p) = \mathbb{E}[f^{z}] + \phi_\rho \sigma[f^z] 
\end{equation}
where $\Phi^{-1}(\rho)$ is associated with the inverse error function $\Phi^{-1}(\rho) = \sqrt{2} \text{erf}^{-1}(2\rho-1)$ such that $f^z$ exceeds $\mathbb{E}[f^{z}] + \phi_\rho \sigma[f^z]$ with probability $1-\rho$ and lie outside $\mathbb{E}[f^{z}] \pm \phi_\rho \sigma[f^z] $ with probability $2(1-\rho)$.
Therefore, by guaranteeing boundedness of the GP-BaS resulted from the dynamics upper bound ${\mathbb{E}}[f^{z}] + \phi_\rho \sigma[f^{z}]$ we can guarantee safety with high probability. This is proven by the following Lemma.
\begin{lemma} \label{lemma: boundedness of bas through gp-bas}
Define $\widetilde{f^z} := {\mathbb{E}}[f^{z}] + \phi_\rho \sigma[f^{z}]$ and therefore with a probability $\rho$ we have that $f^z \leq \widetilde{f^z} \ \forall t$. Let $\widetilde{z}$ be the fictitious GP-BaS that evolves according to the dynamics $\widetilde{f^z}$, i.e. $\widetilde{z}(t) = \int_{0}^{t} \widetilde{f^z}(\tau) d\tau$ with $\widetilde{z}(0) = \beta(x(0)) - \beta(0), \ x(0) \in \mathcal{S}$. Moreover, assume that the resulting $\widetilde{z}$ is rendered bounded. Then, the true barrier state $z$ is bounded with probability $\rho$ for all $t$. 
\end{lemma}
\begin{proof}
With a probability $\rho$, we have 
\begin{align}
 f^z(t) \leq \widetilde{f^z}(t) \ \forall t 
 \Rightarrow & \int_{0}^{t} f^z(\tau) d\tau \leq \int_{0}^{t} \widetilde{f^z}(\tau) d\tau \ \forall t
\end{align}
Since $\tilde{z}(t) = \int_{0}^{t} \widetilde{f^z}(\tau) d\tau$ is bounded, then 
\begin{align}
 \int_{0}^{t} f^z(\tau) d\tau \leq \tilde{z}(t) <\infty \ \forall t \Rightarrow z(t) < \infty \ \forall t
\end{align}
with a probability $\rho$.
\end{proof}
Now, we are in position to develop the safety embedded Gaussian Process dynamical model (GPDM) used to design the safe feedback controller:
\begin{equation} \label{eq:safety embedded Gaussian system}
    \bar{f} := \mathcal{GP}\big(\mathbb{E}[\bar{x}], \mathbb{COV}[\bar{x}] \big)
\end{equation}
where $\bar{x} = \begin{bmatrix} x & z \end{bmatrix}^{\mathsf{T}} \in \bar{\mathcal{X}} \subseteq \mathcal{X} \times \mathcal{Z}$ and $\bar{f} = \begin{bmatrix} f & f^z \end{bmatrix}^{\mathsf{T}} :   \bar{\mathcal{X}} \times \mathcal{U} \rightarrow  \bar{\mathcal{X}}$.

For the safety embedded GPDM, using the aforementioned Lemmas, we propose the main Theorem.
\begin{theorem} \label{theorem: safe and stable if GP embedded system is}
Consider the unknown safety-critical dynamical system in \eqref{eq:dynamics} modeled by the Gaussian processes \eqref{eq:system's Gaussian model} and the barrier states defined by \eqref{eq:zdot-barrier state} modeled with the GP-BaS in \eqref{eq:E[fz]}-\eqref{eq:Var[fz]}. Assume there exists a continuous feedback controller $u=K(\bar{x})$ such that the trajectories of the safety embedded Gaussian system \eqref{eq:safety embedded Gaussian system} are bounded. Then, \textbf{with probability $\rho$}, the trajectories of the safety-critical system \eqref{eq:dynamics} are safe. 
\end{theorem}
\begin{proof}
The results follow directly from \autoref{prop: safety if bas is bounded}, \autoref{theorem: original safe if stable theorem}, \autoref{lemma: gaussian bas lema} and \autoref{lemma: boundedness of bas through gp-bas}. 
\end{proof}

\noindent\textbf{Remark.} An interesting result of the prior development is the fact that the controller does not need to asympotatically stabilize the barrier state $z$. A controller that guarantees \textit{boundedness} of the GP-BaS dynamics will be enough to ensure safety of the unmodeled system under the defined probability.

\noindent\textbf{Remark.} Notice that the results are provided for the general Gaussian dynamics \eqref{eq:safety embedded Gaussian system}. When computing controls that utilize the GP-BaS, the GPDM must be propagated with its uncertainty to meet the safety guarantees proposed in this work.





Finally, we provide the gradients of the GP-BaS system to be utilized in gradient-based control schemes such as the linear quadratic regulator (LQR) and differential dynamic programming (DDP). For linearization of the GP model, the reader may refer to the literature, e.g. \cite{deisenroth2013gaussian,pan2015data} and there references therein. Therefore, we can construct the embedded gradients used in the LQR design for the stabilization problem according to
\begin{gather} \label{eq: augmented jacobian - fx}
    \bar{f}_{\bar{x}}=\begin{bmatrix} f^*_x & 0_{n \times 1} \\ B' \circ B^{-1}(\beta_0) h_x(0) f^*_x -\gamma B'(h_x(0)) & -\gamma \end{bmatrix}  \\ \label{eq: augmented jacobian - fu}
    \bar{f}_{u}=\begin{bmatrix} f^*_u \\ B' \circ B^{-1}(\beta_0) h_x(0)  f^*_u \end{bmatrix}
\end{gather}
where $f^*_x$ and $f^*_u$ are the gradients of the GPs predictive mean which can be computed numerically or analytically as derived in \cite{deisenroth2013gaussian}. We provide rough algorithms, Algorithm \ref{Algorithm: GP-BaS LQR} and Algorithm \ref{Algorithm: GP-BaS DDP} for the control design process used in this paper. The dynamics and gradients of barrier states for the discrete time formulation used for trajectory optimization in this paper are derived in the Appendix \ref{Section: Appendix}.

\begin{algorithm} [h] 
\SetAlgoLined
 \KwData{Training pairs $(\hat{X}, F)$}
 \KwResult{Safe control $u=K(\bar{x})$}
Train $n$ GP models and optimize hyperparamters\;
Select safety probability $p$ and get quantile $\phi_p$\; 
Linearize GP dynamics around nominal trajectory\;
 Obtain GP-BaS Jacobian matrices $f^z_x$, $f^z_u$ and augment to linearized GP dynamics \;
 Use augmented gradients \eqref{eq: augmented jacobian - fx}-\eqref{eq: augmented jacobian - fu} to compute LQR gains\;
\caption{GP-BaS Stabilization with LQR} \label{Algorithm: GP-BaS LQR}
\end{algorithm}
\vspace{-4mm}
\begin{algorithm} [h] 
\SetAlgoLined
 \KwData{Training pairs $(\hat{X}, F)$}
 \KwResult{Safe control $u=K(\bar{x})$}
Train $n$ GP models and optimize hyperparamters\;
Select safety probability $p$ and get quantile $\phi_p$\; 
Define nominal input trajectory and propagate GP dynamics\;
\While{$\Delta V > \epsilon$}{
 Compute running and terminal costs at $k=N$\;
 \For{$k=N-1$ to $1$}{
 Linearize GP dynamics around nominal point $(\tilde{x}_k, \tilde{u}_k)$\;
 Obtain GP-BaS Jacobian matrices $f^z_x$, $f^z_u$ at $(\tilde{x}_k, \tilde{u}_k)$ and augment to linearized GP dynamics\;
 Compute cost quadratization matrices\;
 Compute feed-forward and feedback local optimal variation\;
 }
 \For{$k=1$ to N-1}{
 Forward propagate safety embedded GPDM \eqref{eq:safety embedded Gaussian system} with 
  Backtracking line-search}
 Update $\Delta V$, nominal $\tilde{\bar{x}}$ and nominal $\tilde{u}$
 }
 Apply $u= u_{ff} + u_{fb}(\bar{x})$ through propagating GP-BaS to get optimized trajectory.
\caption{GP-BaS Trajectory Optimization with DDP} \label{Algorithm: GP-BaS DDP}
\end{algorithm}

\section{Applications and Examples} \label{Section: Applications and Examples}
The following examples were implemented utilizing PyTorch \cite{paszke2019pytorch} and GPyTorch \cite{gardner2018gpytorch}. Specifically, the Gaussian process models utilized are task independent variational Gaussian processes. The approximate nature of these models allows us to scale to large amounts of training data, and we ignore coupling between states when learning multi-dimensional models. The mean function utilized in all methods was the constant mean function, and we utilized the squared exponential kernel with automatic relevance detection for our covariance models.

\subsection{Safe Control of Unmodeled Systems}

\subsubsection{Optimal Safe Control of Unstable Linear System}
As a proof of concept, we start with an unstable linear system that needs to be safely and optimally stabilized. The optimal control problem seeks a safe feedback controller that minimizes
\begin{align*}
& V(x(0),u(t))= \int_{0}^{\infty} 0.1 (x_1^2+x_2^2) + 0.01 u^2\; dt \\
& \text{subject to }  \begin{bmatrix} \dot{x}_1 \\ \dot{x}_2 \end{bmatrix}= \begin{bmatrix} 1 & -5 \\ 0 & -1\end{bmatrix} \begin{bmatrix} x_1 \\ x_2 \end{bmatrix} + \begin{bmatrix} 0 \\ 1\end{bmatrix} u \\
& \hspace{1.5cm} (x_1 - 2)^2 + (x_2 - 2.2)^2 - 1 >0
\end{align*}
Notice that the model is unstable and is completely unknown in which we need to learn $2$ independent GP models, one for each state. Note that the optimal controller is now a function of the barrier states as well, which is propagated at each time according to its Gaussian model. The proposed framework was used with a  linear quadratic regulator (LQR) to solve the problem. \autoref{fig:linear system} shows the effectiveness of the proposed approach in generating safe trajectories using only $125$ data points in the state-input space generated randomly in a range of $(-10,10)$ for each GP input. It is also shown that using the covariance in propagating the BaS truly generates a more conservative solution as we should expect. This is also shown in the BaS progression in the bottom figure in which the BaS with covariance obtains higher values. 

\begin{figure} [tbh]
        \centering
    \subfloat{\includegraphics[trim=30 0 50 50, clip, width=0.55\linewidth]{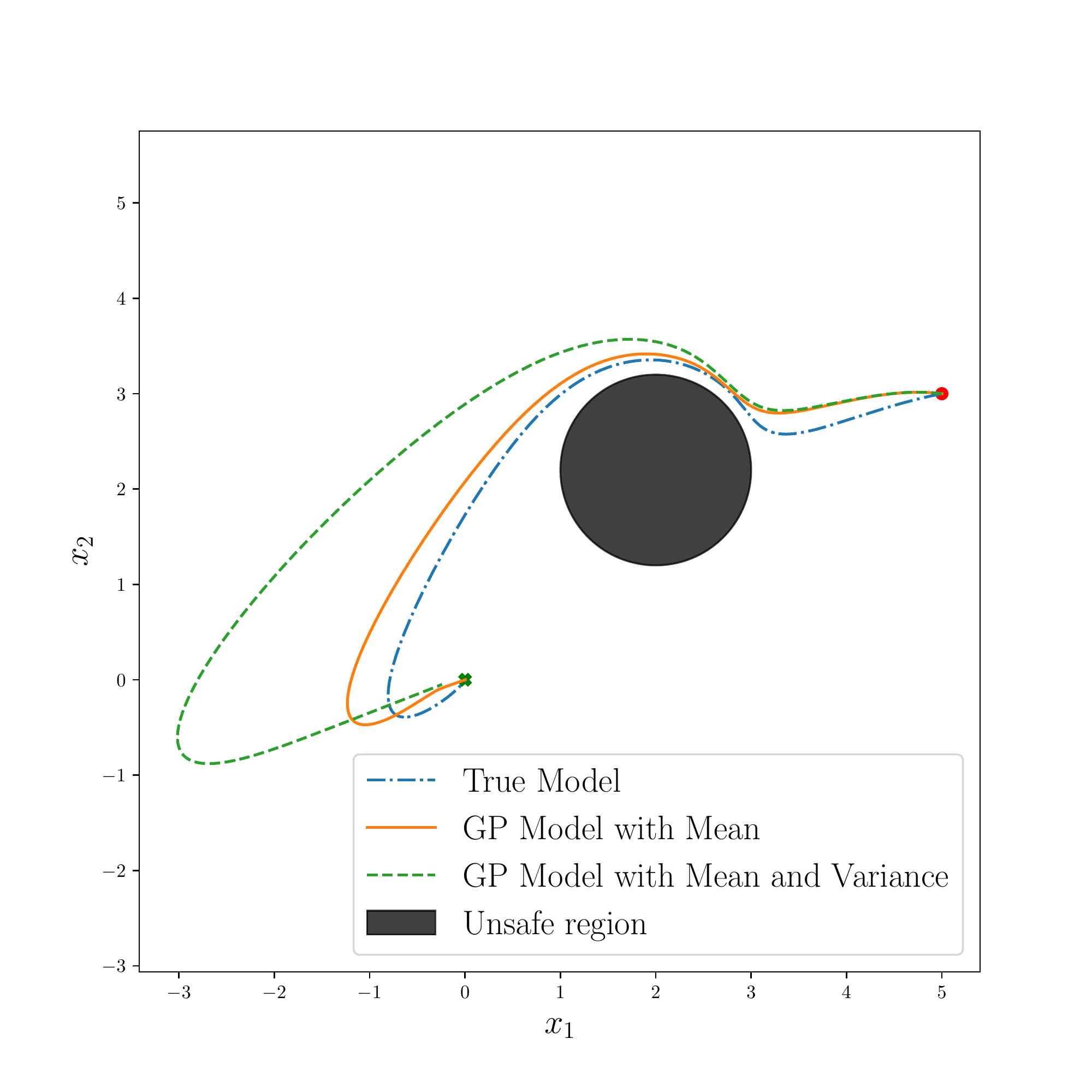}}
    \\ \vspace{-8mm}
    \subfloat{\includegraphics[trim=0 0 0 0, clip, width=.8\linewidth]{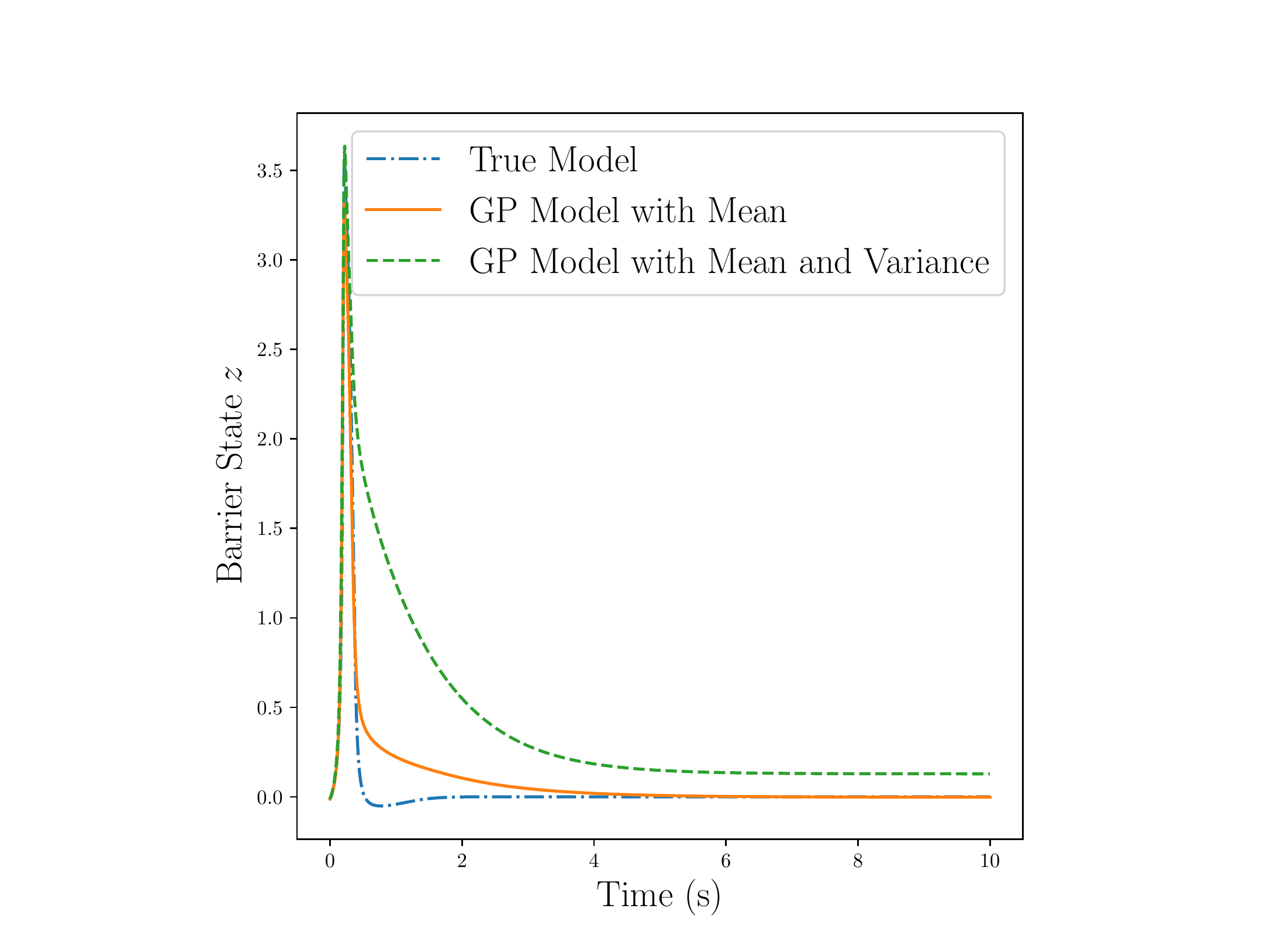}}
    \vspace{-4mm}
   \caption{Closed-loop trajectories of the open-loop unstable linear system under BaS-LQR using true model, GP-BaS-LQR using only GP-BaS mean and using GP-BaS mean and one sigma. It can be seen that the solution is more conservative when the barrier states is propagated with the variance to guarantee safety. The bottom figure shows the progression of the corresponding barrier states.}
      \label{fig:linear system}
      \vspace{-5mm}
\end{figure} 

\subsection{Safe Trajectory Optimization of Unmodeled Systems} 
In this section, we demonstrate how the GP-BaS model can be utilized in conjunction with trajectory optimization to generate safe controls for an unknown system. The method utilized in this section is GPDDP \cite{pan2015data} with discrete barrier states (DBaS) \cite{almubarak2021safeddp}. In the examples, we utilize quadratic cost functions to penalize the systems' states, the barrier states and the control inputs. The basic problem formulation is shown below.
\begin{align*}
\min_{u} J(x,u)  &= \sum^{N-1}_{k=0} \frac{1}{2}\left( x_k^{\text{T}} Q x_k + u_k^{\text{T}} R u_k + x_N^{\text{T}} \phi X_N \right) \\
\textrm{s.t.} \quad & x_{k+1} = f(x_k,u_k)\\
  & h(x_k) > 0 \forall k \in [0, N]
\end{align*}

\subsubsection{Dubins Vehicle}
The system dynamics are given by $\dot x = r\cos\theta(u_1 + u_2) / 2, \dot y = r\sin\theta(u_1 + u_2) / 2, \dot\theta = \frac{r}{2d}(u_1 - u_2)$, where $x$ and $y$ are the robot's coordinates, $\theta$ is its heading, $r=0.2$ is the wheel radius, $d=0.2$ is the wheelbase width, and $u_1$ and $u_2$ are the speeds of the right and left wheels respectively. Notice that this system has an ill-conditioned relative degree which prevents the direct use of CBFs. To learn the vehicle's model, we generated $296$ state-input data points from $4$ sample sinusoidal trajectories of $x$ and $y$. \autoref{fig:dubins 1 obst} shows the success of the proposed frame work to solve the optimization problem using the GP model while maintaining safety and producing a very close solution to the true one. In a more complicated environment, \autoref{fig:dubins- many obst} shows another scenario in which the true model can produce a solution that narrowly passes between the obstacles. It can be seen that the proposed approach could not produce such a solution but generated a more conservative solution to guarantee safety given the uncertainty in the model.

\begin{figure}
        \centering
        \hspace{-2mm}
    \subfloat{\includegraphics[trim=60 0 60 0, clip, width=0.5\linewidth]{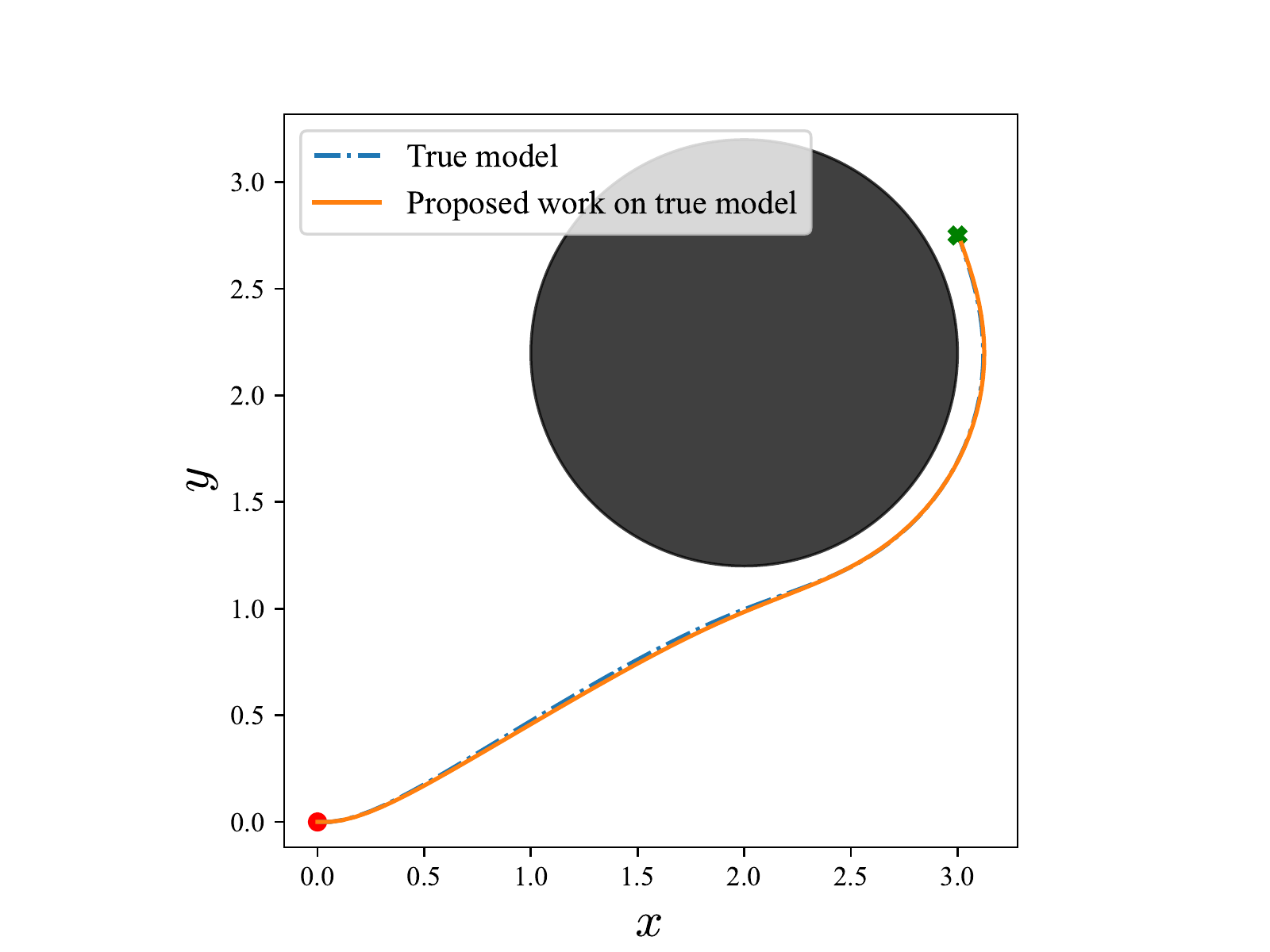}}
    \subfloat{\includegraphics[trim=60 0 60 0, clip, width=.5\linewidth]{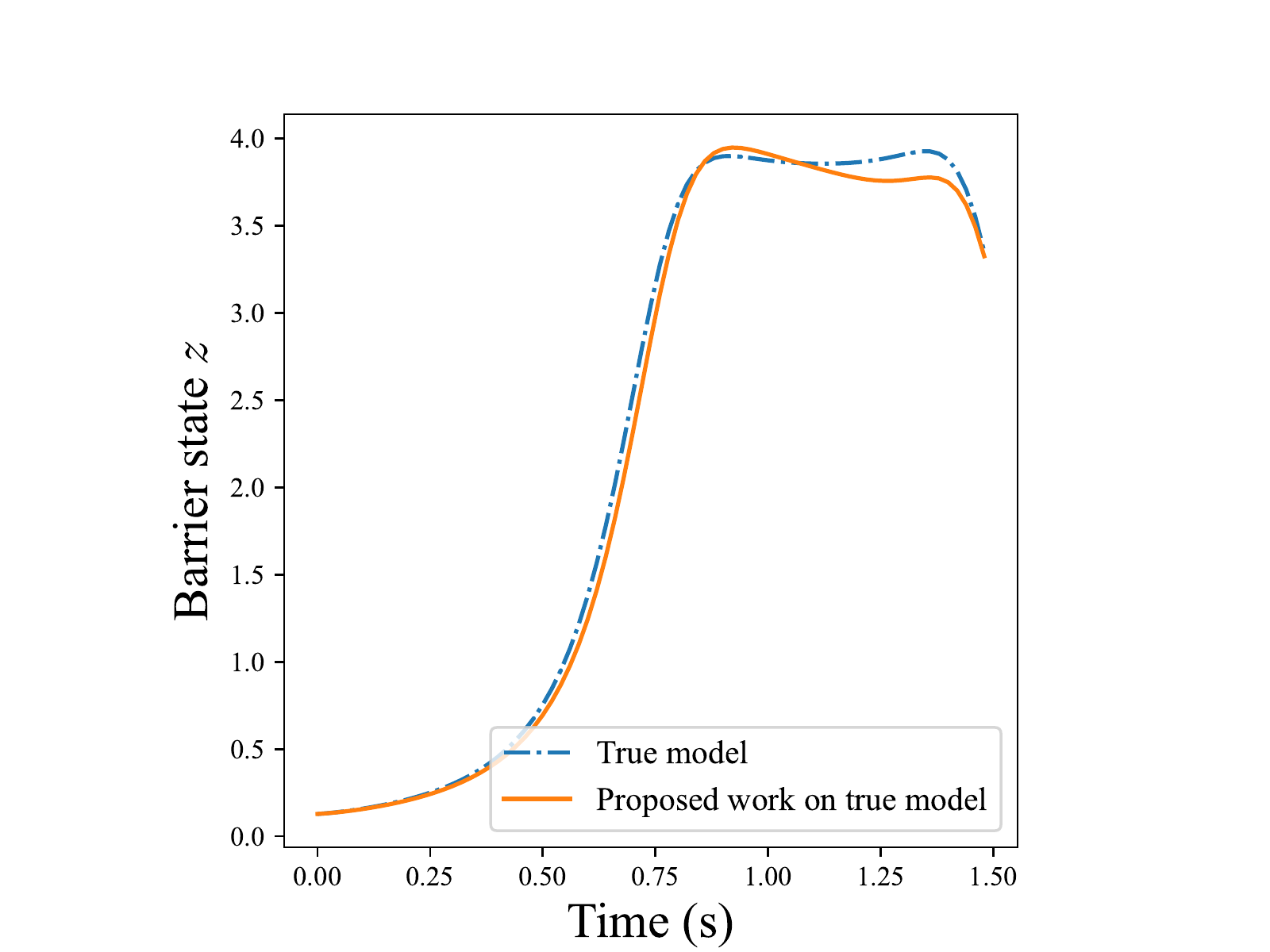}}
   \caption{Closed-loop trajectories of (blue) true quadrotor solution and (orange) GP quadrotor solution. The GP-BaS solution closely estimate the true BaS-DDP solution in such a simple avoidance problem.}
      \label{fig:dubins 1 obst}
\end{figure}

\subsubsection{Quadrotor}
The $12$ state quadrotor system that was utilized for these experiments can be found in \cite{Sabatino2015QuadrotorCM}. Since the kinematics for this model are known, we focus on the dynamics of the model, specifically learning the mapping from input torques and thrust to body accelerations. To learn this model, trajectories were sampled from the true dynamics in an obstacle free scenario. We utilized 7512 training points with 16 induced points in the approximate kernel to capture the dynamics. \autoref{fig:quadrotor} shows successful solution comparing planning obtained using barrier states with DDP on the true model against GP-BaS on the GP model. Clearly, the proposed approach effectively performs the task.

\begin{figure}
        \centering
                \vspace{-18mm}
    \subfloat{\includegraphics[trim=50 90 50 80, clip, width=0.85\linewidth]{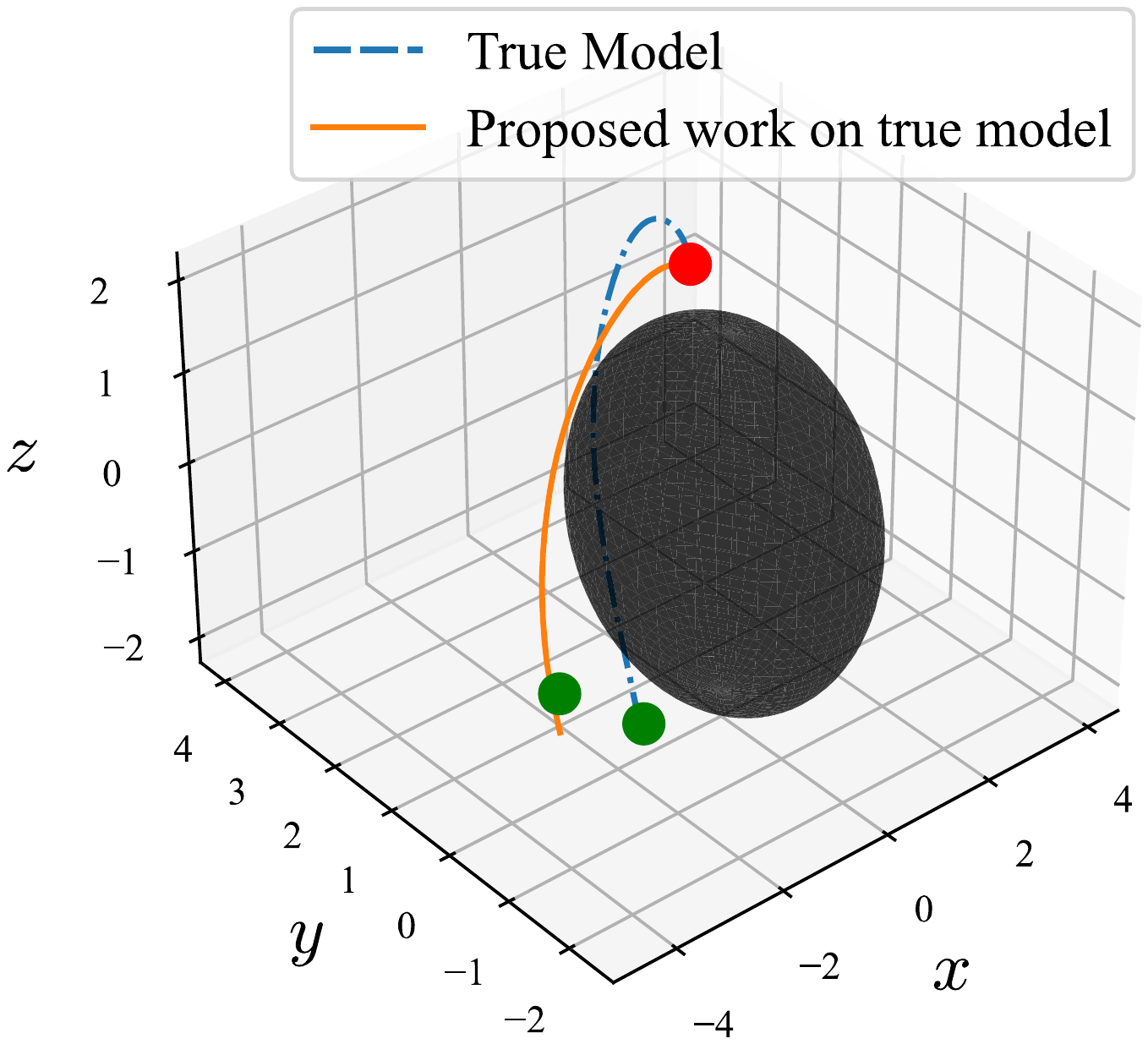}}
        \\ \vspace{-27mm}
    \subfloat{\includegraphics[trim=30 50 50 120, clip, width=0.55\linewidth]{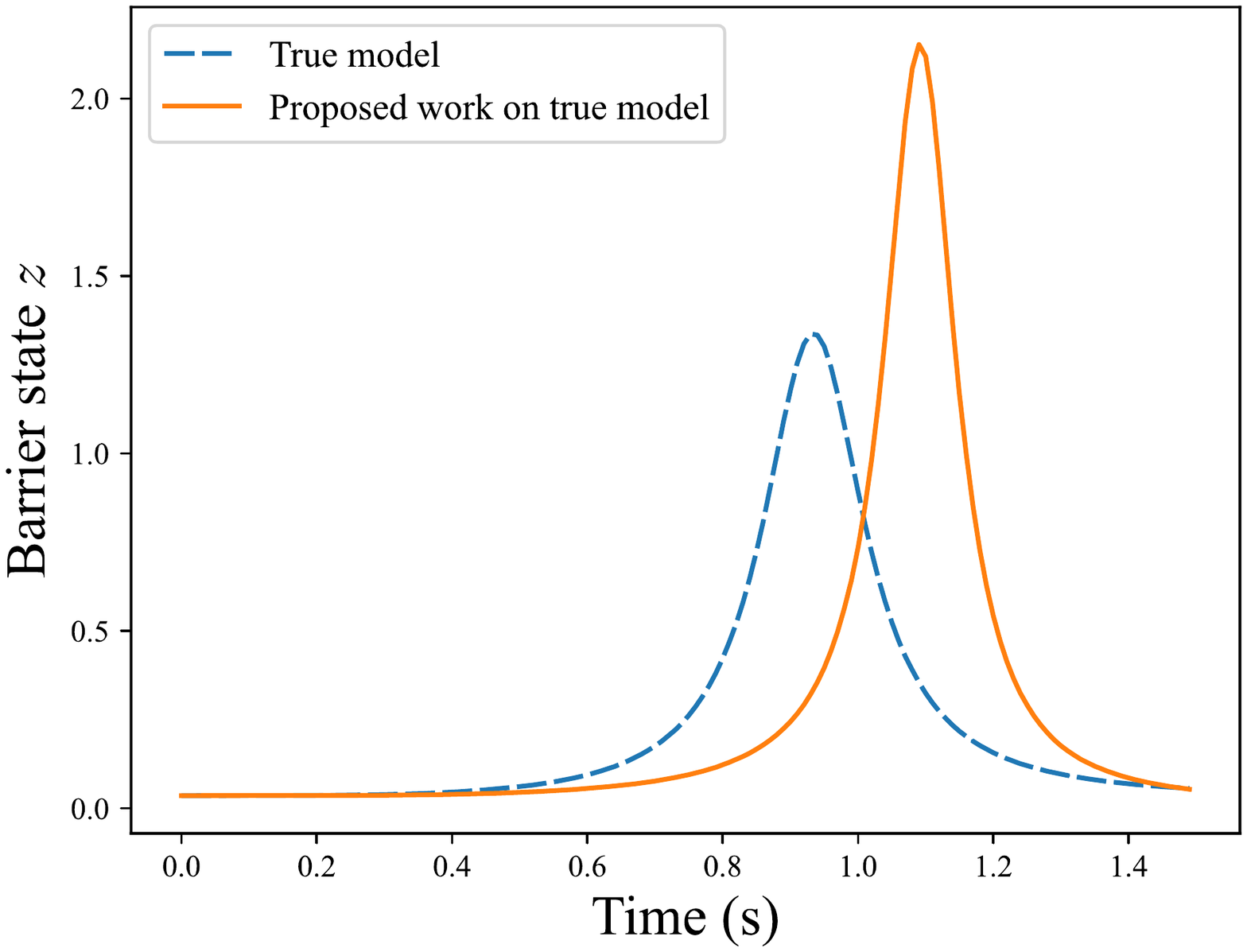}}
    \vspace{-12mm}
   \caption{Closed-loop trajectories of (blue) true quadrotor solution and (orange) GP quadrotor solution. To demonstrate safety with learned GP dynamics, a Gaussian Process was trained and then implemented utilizing GP-BaS. The resultant solution that was computed using the GP dynamics was then tested on the true quadrotor dynamics. It can be seen that the BaS takes large values to push the quadrotor away from the obstacle.}
      \label{fig:quadrotor}
      \vspace{-4mm}
\end{figure} 

\section{Conclusion} \label{Section: Conclusion}
In this work we developed a novel scheme that combines Bayesian learning with barrier states control (GP-BaS). This novel approach circumvents the need for assumptions or special care for the relative degree of the safe set function. Additionally we relax assumptions on the control affinity of the unknown dynamics. We derived probabilistic guarantees of safe control and safe trajectory optimization for unknown systems. To show the efficacy of the proposed work, we considered an optimal safe control design problem and safe trajectory optimization problems of unknown linear and complex nonlinear systems. Future work includes the implementation of this work for online learning, model predictive control, and Probabilistic Differential Dynamic Programming {\cite{Pan2018}}.

    \section{Appendix}
\label{Section: Appendix}
\subsection{Discrete Barrier States}
For the discrete system $f(x_k, u_k)$, the discrete barrier state (DBaS) is given by:
\begin{align}
    w_{k+1} = B \circ h \circ f(x_k, u_k) - \gamma \big(w_k - B \circ h (x_k) \big)
\end{align}
The gradient w.r.t the embedded state (around $\tilde{x}$) is:
\begin{align}
    w_{{k+1}_{\bar{x}}} &= \Big( B'\big(h(\tilde{x}_{k+1})\big) h_x (\tilde{x}_{k+1}) f_x(\tilde{x}_k) \\&+ \nonumber
    \gamma B'\big(h(\tilde{x}_k)\big) h_x (\tilde{x}_k) \Big) \delta x_k - \gamma \delta w_k
\end{align}
where $\tilde{x}_{k+1}= f (\tilde{x}_k, \tilde{u}_k)$. Note that for the Euler \textit{discretized} system $x_{k+1} = x_k + dt f(x_k, u_k)$, the DBaS is given by
\begin{align}
    w_{k+1} &= B \circ h \circ \big(x_k + dt f(x_k,u_k)\big) \\ \nonumber
    &- \gamma \big(w_k - B \circ h (x_k) \big)
\end{align}
Therefore, the gradient w.r.t the embedded state $\bar{x}$ (linearizing around $\tilde{x}$) is:
\begin{align}
    \frac{\partial w_{k+1}}{\partial \bar{x}} &= \Big( B'\big(h(\tilde{x}_{k+1})\big) h_x (\tilde{x}_{k+1}) f^d_x(\tilde{x}_k) \\ \nonumber
    &+ \gamma B'\big(h(\tilde{x}_k)\big) h_x (\tilde{x}_k) \Big) \delta x_k - \gamma \delta w_k
\end{align}
where $\tilde{x}_{k+1}=\tilde{x}_k+dt f (\tilde{x}_k, \tilde{u}_k)$ and $f^d_x(\tilde{x}_k) = \big(I + dt f_x(\tilde{x}_k)\big)$ is the discretized gradient. Note that the gradients of of $B \circ h$ appearing in both terms is different due to the different arguments ($x_{k+1}$ in the first and $x_k$ in the second).

\bibliographystyle{IEEEtranN}
\bibliography{IEEEabrv,references}

\end{document}